\documentclass[12pt,a4paper]{article}
\usepackage{amsmath,amsfonts,amssymb}
\usepackage{array}
\usepackage[caption=false,font=normalsize,labelfont=sf,textfont=sf]{subfig}
\usepackage{booktabs,caption}
\usepackage{relsize}
\usepackage[linesnumbered,ruled,vlined]{algorithm2e}
\usepackage{makecell}
\usepackage{amsthm}
\usepackage{textcomp}
\usepackage{stfloats}
\usepackage{url}
\usepackage{bm}

\usepackage{xcolor}
\usepackage{verbatim}
\usepackage{graphicx}
\usepackage{scalerel}
\usepackage{romannum}

\usepackage[left=2.00cm, right=2.00cm, top=2.00cm, bottom=2.00cm]{geometry}

\SetCommentSty{mycommfont}

\SetKwInput{KwInput}{Input}                
\SetKwInput{KwOutput}{Output}              
\newcommand\bigzero{\makebox(0,0,0){{\Huge 0}}}
\def\BibTeX{{\rm B\kern-.05em{\sc i\kern-.025em b}\kern-.08em
		T\kern-.1667em\lower.7ex\hbox{E}\kern-.125emX}}
\usepackage{balance}
\usepackage{lipsum}
\usepackage{orcidlink}
\usepackage[usestackEOL]{stackengine}
\usepackage[nottoc,notlot,notlof]{tocbibind}
\usepackage[utf8]{inputenc}
\newtheorem{theorem}{Theorem}
\newtheorem{corollary}{Corollary}[theorem]
\newtheorem{lemma}[theorem]{Lemma}
\newtheorem{defn}{Definition}

\newtheorem{exmp}{Example}
\newtheorem{proposition}{Proposition}
\newtheorem*{remark}{Remark}
\date{}
\begin{document}
	\title{ACD codes over skew-symmetric dualities}
	\author{Astha Agrawal\orcidlink{0000-0002-4583-5613} and R. K. Sharma\orcidlink{0000-0001-5666-4103}}
	\markboth{}%
	{ACD codes over skew-symmetric dualities}
	\maketitle
\vspace{-12mm}
\begin{center}
		\noindent {\small Department of Mathematics, Indian Institute of Technology Delhi, New Delhi 110016, India.}
\end{center}
\footnotetext[1]{{\em E-mail addresses:} \url{asthaagrawaliitd@gmail.com} (A. Agrawal), \url{rksharmaiitd@gmail.com} (R. K. Sharma).}
	\begin{abstract}
		The  applications of additive codes mainly lie in quantum error correction and quantum computing. Due to their applications in quantum codes, additive codes have	grown in importance.  In addition to this, additive codes allow the implementation of a variety of dualities. The article begins by developing the properties of Additive Complementary Dual  (ACD) codes  with respect to arbitrary dualities over finite abelian groups. Further, we introduce a subclass of non-symmetric dualities referred to as skew-symmetric dualities. Then, we  precisely count  symmetric and skew-symmetric dualities over finite fields. Two conditions have been obtained: one is a necessary and sufficient condition,  and the other  is a  necessary condition. The necessary and sufficient condition is for an additive code to be an ACD code  over arbitrary dualities. The  necessary condition is on the generator matrix of an ACD code over skew-symmetric dualities. We provide bounds for the highest possible minimum distance of ACD codes over skew-symmetric dualities. Finally, we find some new quaternary ACD codes over non-symmetric dualities with better parameters than the symmetric ones.
	\end{abstract}

	\section{Introduction}
	Additive  codes are a broad class of codes. These codes are defined by additive groups. Additive codes are a more generalised version of linear codes, i.e., all linear codes are additive codes, but the converse need not be true; see \cite{codingtheory}. For the first time in 1973, Delsarte et al. defined additive codes in terms of association schemes, and with the help of additive codes, they gave constructions of schemes with good performance (see, for instance, \cite{algebraicapproach,associationschemes}). Additive codes have gained significance as a topic in algebraic coding theory due to their usage in quantum error correction and quantum computing \cite{quantum}.  The trace inner product, for instance, has been used to explore additive codes over $ \mathbb{F}_4$. Building quantum codes from additive codes has received a lot of attention in the literature \cite{AsymmetricQuantumCodes,quantumcodesfromselfdual,Quantumstabilizercode}. In addition to their usage in quantum computing, additive codes permit the implementation of a range of dualities \cite{Orthogonality}. For an additive code, the product of the cardinalities of the code and its dual with respect to any duality must be equal to the cardinality of the ambient space \cite{Algebraiccoding}. Dougherty et al. \cite{selfdualadditive}  defined  self-dual additive code over a finite abelian group $ \mathbb{Z}_{p^e} $ in terms of an arbitrary duality. They showed the existence of self-dual additive codes under any duality  over a finite abelian group $\mathbb{Z}_{p^e}$. In 2021,  Dougherty and Fernández-Córdoba \cite{AdditiveGcode} defined additive $ G$-codes over finite fields and  proved that if $ M $ and $ M' $ are two dualities, then $ C^M $ and $ C^{M'} $ are equivalent codes. For a variety of dualities, they demonstrated the existence of self-dual codes and established a connection between them and formally self-dual and linear self-dual codes.\\
	Linear complementary dual codes were defined by Messay \cite{lcd} in 1992. He gave a necessary and sufficient condition  for a linear code to be an LCD code. Like linear complementary dual codes (LCD for short), Dougherty et al. \cite{ACDgroupcharacter} defined the notion of additive complementary dual codes (ACD for short) for a variety of dualities. This work was focused on quaternary ACD codes over symmetric dualities.  During our search in MAGMA, we identified improved parameters for quaternary ACD codes over non-symmetric dualities. This prompted our interest in studying non-symmetric dualities. In this article, we introduce a subclass of non-symmetric dualities known as skew-symmetric dualities. We give a necessary and sufficient condition for an additive code to be an ACD code over finite fields for arbitrary dualities.
	
	The article is structured as follows: Section \ref{Preliminaries} provides the necessary definitions and notations used throughout the paper. Section \ref{ACD codes over arbitrary duality} presents results on ACD codes for arbitrary dualities over finite abelian groups and demonstrates the existence of non-trivial ACD codes of length $n$ over $\mathbb{F}_{p^e}$, where $e>2$. In Section \ref{skew-symmetric dualities}, skew-symmetric dualities are introduced. We also give an exact count of symmetric and skew-symmetric dualities over  finite fields. Section \ref{ACD code on skew-symmetric dualities} gives a necessary and sufficient condition for an additive code to be an ACD code over a finite field for any duality. In Section \ref{bounds on skew-symmetric dualities}, some bounds are derived on the highest possible minimum distance for a given length and size of an ACD code for skew-symmetric dualities.  Finally, in Section \ref{non-symmetric dualities}, we find new quaternary codes under non-symmetric dualities with better parameters than the symmetric ones.
	\section{Preliminaries}\label{Preliminaries}
	Let $\mathbb{F}_{q}$ be a finite field of order $q=p^e$, where $p$ is a prime and $e>0$. The additive group  $(\mathbb{F}_{q},+)$ is isomorphic to $\underbrace{\mathbb{Z}_p\times\mathbb{Z}_p\times\dots\times\mathbb{Z}_p}_\text{$e$ times}$.
	\begin{defn}
		An additive code $C$ of length $n$ is defined as an additive subgroup of  $G^n$, where $G$ is an  abelian group of finite order.
	\end{defn} 
	Let $C\subseteq \mathbb{F}_q^n$, then $C$ is a linear code if it is a subspace of $\mathbb{F}_q^n$, here $\mathbb{F}_q^n$ is considered as $n$ dimensional vector space over $\mathbb{F}_q$. On the other side, additive codes over $\mathbb{F}_q$ are not necessary subspaces rather they are additive subgroups of  an additive group  $(\mathbb{F}_{q}^n,+)$. In addition, additive codes of length $n$ are $\mathbb{F}_p$-linear codes over $\mathbb{F}_{p^e}$.
	\begin{defn}
		A generator matrix for an additive code $C$ over $\mathbb{F}_{p^e}$ is a matrix $\mathcal{G}$ that generates the code $C$  by  $\mathbb{F}_p$-linear combinations of  rows of $\mathcal{G}$, i.e., $C$ contains the vector $\sum\alpha_i\mathcal{G}_i$, where $\mathcal{G}_i$ is the $i$-th row of $ \mathcal{G} $ and $\alpha_i\in\mathbb{F}_p$.
	\end{defn}
	\begin{defn}
		The Hamming weight of a vector $\textbf{x}$ is denoted by   $wt(\textbf{x})$ and is equal to the number of non-zero elements in the vector.
	\end{defn}
	 The Hamming weight of $\textbf{x}-\textbf{y}$, or $wt(\textbf{x}-\textbf{y}$), is used to describe the distance between two codewords, $\textbf{x}$ and $\textbf{y}$, in $C$. The distance of an additive code $C$ is defined as $d=\min\{wt(\textbf{x})\mid \textbf{x}\neq0\}$. 
	
	Additive codes allow the range of dualities defined by the character of a finite abelian group $G$. A character of  $G$ is defined as a homomorphism from $G$ to $\mathbb{C}^*,$ the multiplicative group of complex numbers. The collection of all characters on $G$ forms a group that is known as the character group $\hat G$, i.e., $\hat G=\{\chi|\chi\  \text{is a character on G}\}$. It is worth noting that this group $\hat G$ is isomorphic to $G$.  We can define  isomorphism  from $G$ to $\hat G$, which helps to give a character table or a duality $ M $ as follows: Let $\phi$ be an isomorphism from $ G $ to $\hat G$, defined as $\phi(x_i)=\chi_{x_i}$, where $x_i\in G$ and $\chi_{x_i}$ is a character. Then $ij$-th entry of the character table $M$ is $\chi_{x_i}(x_j),\forall x_i,x_j\in G$. With the help of these dualities, we can define orthogonality on additive codes.\\
	If we consider $G=(\mathbb{F}_{p^e},+)$, then a character is defined as  $\chi: \mathbb{F}_{p^e}\to P\subset \mathbb{C}^*$ such that $\chi(x)=\xi$, where $\xi \in P$ and $P$ is the set of all $p$-th roots of unity, which is a cyclic group of order $p$. Let $\textbf{a}=(a_1,a_2,a_3,\dots,a_n)$ and $\textbf{b}=(b_1,b_2,b_3,\dots,b_n)\in \mathbb{F}_q^n$, then $\chi_\textbf{a}(\textbf{b})=\prod_{i=1}^{n}\chi_{a_i}(b_i)$.
	
	\begin{defn}
		Let $C$ be an additive code of length $n$ over a group  $G$, then the orthogonal $C^M$, for a fix duality $ M $, is defined as $C^M=\{(x_1,x_2,\dots,x_n)\in G^n\mid \prod_{i=1}^{n}\chi_{x_i}(c_i)=1, \ \forall (c_1,c_2,\dots,c_n)\in C\}$.
	\end{defn} 
	We can observe that $C^M$ is always an additive code, irrespective of $C$. Additive codes satisfy the cardinality condition for all types of dualities, i.e., $|C||C^M|=|G^n|$ see \cite{Algebraiccoding}. It is important to know that the particular duality of a group  is taken into account when defining this orthogonal. That is, if we change the duality, then  the orthogonal of code may also change. In fact, in one duality, a code can be the same as its dual but not in another.  If $C$ is an additive code over a finite abelian group $ G $ with a fixed duality $ M $, then $(C^M)^M=C$, where $M$ is a symmetric duality, that is,  $M=M^T$ or $\chi_x(y)=\chi_y(x)$, for all $x,y\in G$. However, many dualities are not symmetric, for which $M\neq M^T$ or $\chi_x(y)\neq\chi_y(x)$ for some $x,y\in G$. Hence, we can divide the dualities into two classes, i.e., symmetric  and non-symmetric dualities. 
	
	\begin{exmp}\label{duality over F_4}
		Let $G=(\mathbb{F}_9,+)=\left\langle 1,\omega \right\rangle=\{0,1,2,\omega,\omega+1,\omega+2,2\omega,2\omega+1,2\omega+2\}$ be a finite abelian group. Let's see an example of symmetric and non-symmetric duality over the additive group of $\mathbb{F}_9$. 
			\begin{center}
			\begin{table}[htbp]
			\begin{center}
				\begin{tabular}{|c|c|c|c|c|c|c|c|c|c|}
					\hline
					$ \bm{N_1} $&0&1&2&$ \omega $&$ \omega+1 $&$ \omega+2 $&$ 2\omega $&$ 2\omega+1 $&$ 2\omega+2 $\\
					\hline
					0&1&1&1&1&1&1&1&1&1\\
					\hline
					1&1 &$\alpha$&   1&    $\alpha $ & $\alpha^2$ & $\alpha^2$ &   1&$ \alpha^2$&   $\alpha$\\
					\hline
					2&	1 &   1 &$\alpha^2$ & $\alpha$ & $\alpha$&   1& $\alpha$ &$\alpha^2 $& $\alpha^2$\\
					\hline
					$ \omega $&1  & $\alpha$ & $\alpha$ & 1& $\alpha$& $\alpha^2$ &$\alpha^2$& 1& $\alpha^2$\\
					\hline
					$ \omega+1 $&	1 & $\alpha^2$& $\alpha$& $\alpha$& 1& $\alpha$ &$\alpha^2$& $\alpha^2$ &1\\
					\hline
					$ \omega+2 $&1 & $\alpha^2$ &  1&  $\alpha^2$ & $\alpha$& $\alpha$ &  1& $\alpha$& $\alpha^2$\\
					\hline
					$ 2\omega $&	1 &  1 &$\alpha$ & $\alpha^2$ & $\alpha^2$ &  1& $\alpha^2$& $\alpha$ &$\alpha$\\
					\hline
					$ 2\omega+1 $&1 & $\alpha^2$& $\alpha^2$&  1& $\alpha^2$& $\alpha$& $\alpha$ &1&$\alpha$\\
					\hline
					$ 2\omega+2 $&1  & $\alpha$& $\alpha^2$  &$\alpha^2$ & 1& $\alpha^2$& $\alpha$& $\alpha$ & 1\\
					\hline
				\end{tabular}\vspace{0.5cm}
				
				\begin{tabular}{|c|c|c|c|c|c|c|c|c|c|}
					\hline
					$\bm{N_2}$&$0$&$1$&$2$&$\omega$&$\omega+1$&$\omega+2$&$2\omega$&$2\omega+1$&$2\omega+2$\\
					\hline
					$0$&$1$  &$1$  &$1$  &$1$  & $1$ & $1 $&$1$  &$1 $ &$1$  \\
					\hline
					$1$&	1&   1&   $\alpha$&  $\alpha^2$&  $\alpha^2$&   1& $\alpha^2$&    $\alpha$&   $\alpha$  \\
					\hline
					$2$&	1& $\alpha^2$&   $\alpha$&  $\alpha$&  1& $\alpha$& $\alpha^2$&  $\alpha^2$& 1 \\
					\hline
					$\omega$&1& $\alpha$& 1& $\alpha$& $\alpha^2$& $\alpha^2$& 1& $\alpha^2$& $\alpha$  \\
					\hline
					$\omega+1$&1& $\alpha$& $\alpha$& 1& $\alpha$& $\alpha^2$& $\alpha^2$& 1& $\alpha^2$ \\
				\hline
				$ \omega+2 $&1&   1& $\alpha^2$&  $\alpha$&  $\alpha$&   1& $\alpha$&  $\alpha^2$& $\alpha^2$ \\
				\hline
				$ 2\omega $&1& $\alpha$& $\alpha^2$&  $\alpha^2$& 1& $\alpha^2$& $\alpha$& $\alpha$& 1 \\
				\hline
				$ 2\omega+1 $&	1& $\alpha^2$& 1&  $\alpha^2$& $\alpha$& $\alpha$& 1&  $\alpha$& $\alpha^2$ \\
				\hline
				$ 	2\omega+2 $&	1& $\alpha^2$& $\alpha^2$&  1&  $\alpha^2$& $\alpha$& $\alpha$&  1& $\alpha$\\
				\hline
				\end{tabular}			
					
					\end{center}
			\end{table}
		\end{center}
		In the above table, `$\alpha$' denotes the primitive 3rd root of unity. We can observe that $N_1$ is a symmetric duality, and $N_2$ is a non-symmetric duality over $\mathbb{F}_9$.
	\end{exmp}
	\begin{defn}
		Let $C$ be an additive code over a group $G$. If $C\subseteq C^M$, then $C$ is said to be a self-orthogonal code with respect to duality $M$, and if $C=C^M$, then $C$ is said to be a self-dual code with respect to duality $M$.
	\end{defn}
	\begin{defn}
		An additive code $C$ over a group $G$ is said to be an  additive complementary dual code (ACD) with respect to the duality $M$ if $C\cap C^M=\{0\}$.
	\end{defn}
	The  following lemma presents a common approach for constructing bigger ACD codes by utilizing smaller ACD codes.
	\begin{lemma}\cite{ACDgroupcharacter}\label{cross}
		Let $ C $ and $ D $ be ACD codes of lengths $ n  $ and
		$ m $, respectively, over a group $ G $ with respect to duality $ M $. Then $ C\times D  $ is an ACD code with respect to duality $ M $ of length $ n + m $.
	\end{lemma}
	Now, we recall the following results from \cite{Dualities_AbelianGroup}.
		\begin{theorem}\cite[Theorem 3.1]{Dualities_AbelianGroup}
		Given a finite abelian group $G$ and a duality $M$ on $G$, $M^T $ is also a duality on $G$, where $M^T$ refers to the transpose of the character table corresponding to $M$.  
	\end{theorem}
	\begin{remark}
		Let $\chi$ and $\chi'$ represent the duality $M$ and $M^T$, then we have $\chi'_x(y)=\chi_y(x),$ for all $x,y\in G$.
	\end{remark}
	\begin{theorem}\cite[Theorem 3.3]{Dualities_AbelianGroup}
		Given a finite abelian group $G$, if $C$ is an additive code, then $(C^M)^{M^T}=(C^{M^T})^{M}=C$, for any duality $M$ on $G$.
	\end{theorem}
		\begin{remark}
		For symmetric dualities,  it directly follows that $(C^M)^M=C$.
	\end{remark}
	\section{ACD codes over arbitrary dualities}\label{ACD codes over arbitrary duality}
In this section, we establish some properties of additive codes  over arbitrary dualities. Additionally, we prove that for any duality $M$ over a finite field,  an ACD code of any length $n$ exists.

	\begin{theorem}\label{equality}
		Let $C$ and $C'$ be  additive codes, and $M$ be any duality over the finite abelian group $G$. Then we have the following equalities:	\begin{enumerate}
			\item $ (C+C')^M=C^M\cap C'^M $.\label{equality1}
			\item $ (C\cap C')^M=C^M+C'^M $.\label{equality2}
		\end{enumerate}
	\end{theorem}
	\begin{proof}
		\begin{enumerate}
			\item Let $\textbf{h}=(h_1,h_2,\dots,h_n)\in (C+C')^M$, i.e., $\prod_{i=1}^{n}\chi_{h_i}(d_i)=1, \ \ \forall (d_1,d_2,\dots,d_n)\in C+C'.$ It follows that, $\prod_{i=1}^{n}\chi_{h_i}(c_i+0)=1, \ \ \forall (c_1,c_2,\dots,c_n)\in C$ and $\prod_{i=1}^{n}\chi_{h_i}(0+c'_i)=1, \ \ \forall (c'_1,c'_2,\dots,c'_n)\in C'.$ As a result, $\textbf{h}\in C^M\cap C'^M.$ Conversely, let $\textbf{h}=(h_1,h_2,\dots,h_n)\in C^M\cap C'^M$, i.e., $\prod_{i=1}^{n}\chi_{h_i}(c_i)=1, \ \ \forall (c_1,c_2,\dots,c_n)\in C$ and $\prod_{i=1}^{n}\chi_{h_i}(c'_i)=1, \ \forall (c'_1,c'_2,\dots,c'_n) \in C'.$ Further, $\prod_{i=1}^{n}\chi_{h_i}(c_i)\prod_{i=1}^{n}\chi_{h_i}(c'_i)=1$ which implies that $\prod_{i=1}^{n}\chi_{h_i}(c_i+c'_i)=1$. Hence, $\textbf{h}\in (C+C')^M$. This shows that $ (C+C')^M=C^M\cap C'^M $. 
			\item The first equality is true for all dualities and all additive codes. Thus, replace $M$ with $M^T$ and $C$ with $C^M$ in the first equality. We get, $ (C^M+C'^M)^{M^T}={(C^M)}^{M^T}\cap {(C'^M)}^{M^T}=C\cap C'$. This shows that $ C^M+C'^M =(C\cap C')^M.$
		\end{enumerate}
	\end{proof}

	\begin{proposition}\label{C ACD M and MT}
		An additive code $C$ is an ACD code with respect to duality $M$ if and only if it is also  an ACD  code with respect to $M^T$. 
	\end{proposition}
	\begin{proof}
		If $C$ is an ACD code with respect to duality $M$ then from \ref{equality2} of Theorem \ref{equality}, we have $(C\cap C^{M^T})^M=C^M+C=G^n$. This means that $  C\cap C^{M^T}=\{\textbf{0}\}$. Conversely, if $ C $ is  an ACD code  with respect to $M^T$,  then $ (C\cap C^M )^{M^T}= C^{M^T}+C=G^n$. This implies that $  C\cap C^{M}=\{\textbf{0}\}$.
	\end{proof}
	\begin{proposition}\label{C^{M^T} is ACD}
		If $C$ is an ACD code with respect to duality $M$, then $C^M$  and $C^{M^T}$ are ACD codes with respect to duality $M^T$ and $M$, respectively.
	\end{proposition}
	\begin{proof}
		Let $C$ be an ACD code with respect to duality $M$, i.e., $C\cap C^M=\{\bm{0}\}$. By the above Proposition \ref{C ACD M and MT}, we note that $C\cap C^{M^T}=\{\bm{0}\}$. Then, we have $C^M\cap (C^M)^{M^T}=C^M\cap C=\{\bm{0}\}$ and  $C^{M^T}\cap (C^{M^T})^M=C^{M^T}\cap C=\{\bm{0}\}$. This completes the proof.
	\end{proof}
	\begin{proposition}
		An additive code $C$ is self-dual (self-orthogonal)  with respect to  duality $M$ if and only if it is also self-dual (self-orthogonal)  with respect to  duality $M^T$.
	\end{proposition}
	\begin{proof}
		If $C=C^M$, then it is clear that $C^{M^T}=(C^M)^{M^T}=C.$ Also, if $C\subseteq C^M$, then $C\subseteq C^{M^T}$. Conversely,  	if $C=C^{M^T},$ then it is clear that $C^{M}=(C^{M^T})^{M}=C.$ Also, if $C\subseteq C^{M^T}$, then $C\subseteq C^{M}$.
	\end{proof}
	
	\begin{lemma}\label{ChiACD}
		An additive code $C=\{0,x,2x,3x,\dots,(p-1)x\}$, for some  $x\in\mathbb{F}_{p^e}$ is an ACD code  if and only if   $\chi_x(x)\neq1$. 
	\end{lemma}
	\begin{proof}
		Let $\chi_x(x)=\xi\neq1$, where $\xi$ is any primitive $p$-th root of unity. Let $z\in C\cap C^M$ then $z=mx$, for some $0\leq m\leq (p-1)$ and $\chi_{(mx)}(nx)=1$, for all $0\leq n\leq (p-1).$ Therefore, we have  $\chi_{(mx)}(x)=1$, which gives that $\xi^m=1.$ Thus, $m\equiv 0\mod p$. This shows that $z=0$. Hence, $C$ is an ACD code. Conversely, suppose $\chi_x(x)=1$ then $\chi_{mx}(nx)=(\chi_x(x))^{nm}=1$, for all $0\leq n,m\leq (p-1).$ It is a contradiction for $C$ to be an ACD code.
	\end{proof}
	
	\begin{lemma}\label{chi-x-y-not1}
		For any duality, if $\chi_x(x)=1, \ \forall x \in \mathbb{F}_q$ then $ \chi_y(x)= (\chi_x(y))^{-1}, \ \forall y \in \mathbb{F}_q .$
	\end{lemma}
	\begin{proof}
		Let $x,y\in \mathbb{F}_q$, then $\chi_{x+y}(x+y)=1$, this implies that $\chi_x(y)\chi_y(x)=1.$
	\end{proof} 
	\begin{theorem}\label{C-M-M- equl to C}
		Let $M$ be a duality such that $\chi_x(x)=1,\ \forall x \in \mathbb{F}_{q}$. If $C$ is an additive code over $ \mathbb{F}_{q}$, then $(C^M)^M=C$.
	\end{theorem}
	\begin{proof}
		Let $C$ be an additive code over $\mathbb{F}_q$ and $M$ be a duality such that $\chi_x(x)=1$ for all $x\in \mathbb{F}_q$. We have $(C^M)^M=\{(x_1,x_2,\dots,x_n)\in \mathbb{F}^n_q| \prod_{i=1}^{n}\chi_{x_i}(c_i)=1, \forall (c_1,c_2,\dots,c_n)\in C^M\}$. Let $(x_1,x_2,\dots,x_n)\in (C^M)^M$ if and only if
		\begin{align*}
			& \prod_{i=1}^{n}\chi_{x_i}(c_i)=1, \forall (c_1,c_2,\dots,c_n)\in C^M\\
			\iff &\prod_{i=1}^{n}(\chi_{c_i}(x_i))^{-1}=1, \forall (c_1,c_2,\dots,c_n)\in C^M \ \ 
			\text{(by the above Lemma)}\\
			\iff &\left( \prod_{i=1}^{n}\chi_{c_i}(x_i)\right) ^{-1}=1, \forall (c_1,c_2,\dots,c_n)\in C^M \\
			\iff & \prod_{i=1}^{n}\chi_{c_i}(x_i)=1, \forall (c_1,c_2,\dots,c_n)\in C^M \\
			\iff & \prod_{i=1}^{n}\chi'_{x_i}(c_i)=1, \forall (c_1,c_2,\dots,c_n)\in C^M \ \text{($\chi'$ represents the duality $M^T$)}\\
			\iff & (x_1,x_2,\dots,x_n)\in (C^M)^{M^T}=C.
		\end{align*}
		This completes the proof.
	\end{proof}
	Let $x,y\in\mathbb{F}_q$ and $M$ be a duality over $\mathbb{F}_q$, then $x$ is orthogonal to $y$ if and only if $\chi_x(y)=1$. We say an element $x\in \mathbb{F}_q$ is a self-orthogonal element if $\chi_x(x)=1$.
	\begin{theorem}
		For any duality $M,$ there  exists a non-trivial ACD code of length one over $\mathbb{F}_{p^e
		},\ e>2.$
	\end{theorem}
	\begin{proof}
		Let us consider an element $x\in\mathbb{F}_{p^e}$ such that $x$ is not self-orthogonal, i.e., $\chi_x(x)\neq1$, then  by Lemma \ref{ChiACD}, the code $C=\left\langle x\right\rangle =\{0,x,2x,3x,\dots,(p-1)x\}$ is an ACD code. Now, assume that all the elements of $\mathbb{F}_{p^e}$ are self-orthogonal with respect to duality $M$, i.e., $\chi_x(x)=1, \ \forall x \in \mathbb{F}_{p^e}$. Then take an element $y\in \mathbb{F}_{p^e}$ such that $ \chi_x(y)=\xi\neq1,$ where $\xi$ is any primitive $p$-th root of unity (there exists such $ y $, since a non-zero element can not be orthogonal to the whole space). Construct a code $C=\left\langle x,y \right\rangle=\{nx+my\mid0\leq n,m\leq (p-1)\}.$ Now, we will show that $C$ is an ACD code of length 1.
		Suppose $z\in C\cap C^M$ then $z=n_1x+m_1y$, for some $0\leq n_1,m_1\leq (p-1)$ and $\chi_{(n_1x+m_1y)}(nx+my)=1$, for all $0\leq n,m\leq (p-1).$ Therefore, $\chi_{(n_1x+m_1y)}(x)=(\chi_x(x))^{n_1}
		(\chi_y(x))^{m_1}=(\chi_y(x))^{m_1}=1$.  Also,  $\chi_{(n_1x+m_1y)}(y)=(\chi_x(y))^{n_1}
		(\chi_y(y))^{m_1}=(\chi_x(y))^{n_1}=1$. By  Lemma \ref{chi-x-y-not1}, $ \chi_y(x)=\xi^{-1} $. It follows that  $m_1=0=n_1$.  Hence, $C$ is an ACD code of length 1 over $\mathbb{F}_{p^e}, e>2$.
	\end{proof}
	
	\begin{remark}
		We can observe that in the above theorem, $e$ should be greater than 2 because in  $\mathbb{F}_{p^2}$, the code $C=\left\langle x,y\right\rangle $  is the whole space $\mathbb{F}_{p^2}$, which is not a non-trivial ACD code. 
	\end{remark}
	\begin{corollary}
		For any duality $M,$ there  exists a  non-trivial ACD code of length $ n $ over $\mathbb{F}_{p^e},\ e>2.$
	\end{corollary}	
	\begin{proof}
		By the above theorem, we have a code $ C $ of length one over $\mathbb{F}_{p^e},\ e>2.$ Further, the proof follows from Lemma \ref{cross}.
	\end{proof}
	
	\section{Skew-symmetric dualities }\label{skew-symmetric dualities}
	This section targets to establish the existence of a new class of dualities over $ \mathbb{F}_{p^{2e}} $. The signature property of this class is that every element is self-orthogonal. We begin by providing the exact count of  symmetric dualities over  $ \mathbb{F}_{p^e} $.
	\begin{theorem}\label{symduality}
			The total number of symmetric dualities over $\mathbb{F}_{p^{e}}$ for any $e\geq 1$ is equal to 	\begin{equation*}
			\begin{cases}
				p^{k(k-1)}(p-1)(p^3-1)\dots(p^{2k-1}-1) & \text{if e = 2k-1},\\
				p^{k(k+1)}(p-1)(p^3 - 1)\dots(p^{2k-1}-1)& \text{if e = 2k}.\\
			\end{cases}       
		\end{equation*}	
	\end{theorem}
	\begin{proof}
		Let the generators of $ \mathbb{F}_{p^{e}} $ be $x_1,x_2,\dots,x_{e}$ such that $\mathbb{F}_{p^{e}}=\left\langle  x_1,x_2,\dots,x_{e} \right\rangle$. We define a map $\phi:\mathbb{F}_{p^{e}}\to \widehat{\mathbb{F}_{p^{e}}}$ such that $\phi(\sum_{i=1}^{e}n_ix_i)=\prod_{i=1}^{e}\left( \chi_{x_i}\right)^{n_i}=\chi_{\sum_{i=1}^{e}n_ix_i}$ (say), where $\chi_{x_i}$\textquotesingle s are characters for all $ i=1,2,\dots, e$ such that $\chi_{x_i}(x_j)=\xi^{k_{ij}} \text{\ and\ } \chi_{x_i}(x_j)=\chi_{x_j}(x_i), \ \forall 1\leq i,j\leq e$, where $\xi$ is the primitive $p$-th root of unity and $ 0\leq k_{ij}\leq (p-1) $. Clearly, $\phi$ is a homomorphism.  Let $x\in ker(\phi)$, i.e., $\chi_x=\chi_0$. This gives that $\chi_x(z)=1, \ \forall z\in \mathbb{F}_{p^{e}}$. Since $x\in\mathbb{F}_{p^{e}}=\left\langle x_1,x_2,\dots,x_{e}\right\rangle $ which implies that $\chi_x(z)=\chi_{(\sum_{i=1}^{e}n_ix_i)}(z)=1, \ \forall z\in \mathbb{F}_{p^{e}}.$ Therefore, $\chi_{(\sum_{i=1}^{e}n_ix_i)}(x_j)=1, \  \forall j=1,2,\dots, e$. Now, we have
		\begin{align*}
			\chi_{(\sum_{i=1}^{e}n_ix_i)}(x_j)=&\prod_{i=1}^{e}\left(\chi_{x_i}(x_j)\right)^{n_i}
			=\prod_{i=1}^{e}\xi^{k_{ij}n_i}\\
			=&\xi^{\left( \sum_{i=1}^{e}k_{ij}n_i\right)} =1.
		\end{align*}
		This gives  $  \sum_{i=1}^{e}k_{ij}n_i\equiv 0 \mod p,$ for all $j=1,2,\dots e$. It is equivalent to find a solution of the following system of equations\\
			$$\begin{bmatrix}
			k_{11}&k_{12}&k_{13}&\dots&k_{1e}\\
			k_{12}&k_{22}&k_{23}&\dots&k_{2e}\\
			k_{13}&k_{23}&k_{33}&\dots&k_{3e}\\
			\vdots&\vdots&\vdots&\ddots&\vdots\\
			k_{1e}&k_{2e}&k_{3e}&\dots&k_{ee}
		\end{bmatrix}_{e\times e}  \begin{bmatrix}
			n_1\\n_2\\n_3\\ \vdots\\n_{e}
		\end{bmatrix}\equiv\begin{bmatrix}
			0\\0\\0\\ \vdots\\0
		\end{bmatrix}\mod p.$$
			The system of equations has only a trivial solution under mod p if and only if the above symmetric matrix is invertible over $\mathbb{F}_p$. For invertible matrices, if $x\in ker(\phi)$ then $x=0$, which shows that $\phi$ is an isomorphism.
		Hence, $\phi$ gives a symmetric duality if and only if the above $e\times e$ symmetric matrix is invertible  over $\mathbb{F}_p$. The number of $e\times e$ invertible symmetric matrices over $\mathbb{F}_p$ are (see \cite[Exercise 198]{invertiblesymmetric})
		\begin{equation}\label{sym}
			\begin{cases}
				p^{k(k-1)}(p-1)(p^3-1)\dots(p^{2k-1}-1) & \text{if e = 2k-1},\\
				p^{k(k+1)}(p-1)(p^3 - 1)\dots(p^{2k-1}-1)& \text{if e = 2k}.\\
			\end{cases}       
		\end{equation}
		This completes the proof.	
		
	\end{proof}

	\begin{exmp}
		Let $p=3$ and $e=2$, then there are total $(3^2-1)(3^2-3)=48$ dualities over $\mathbb{F}_{3^2}$. In which, $3^{1(1+1)}(3-1)=18$ are symmetric dualities  and 30 are non-symmetric dualities over $\mathbb{F}_{3^2}$.
	\end{exmp}
	It is quite interesting to study the class of non-symmetric dualities since it is such a huge class of dualities.   First of all, in the following theorem, we show the existence of a non-symmetric duality over $\mathbb{F}_{p^{2e}}$ such that every element is self-orthogonal. 
	\begin{theorem}
		Over $\mathbb{F}_{p^{2e}}, \ p>2$ and for any $e\geq 1$, there exist non-symmetric dualities  such that $\chi_{x}(x)=1, \ \forall x \in \mathbb{F}_{p^{2e}}$.
	\end{theorem}
	\begin{proof}
		Let the generators of $ \mathbb{F}_{p^{2e}} $ be $x_1,x_2,\dots,x_{2e}$ such that $\mathbb{F}_{p^{2e}}=\left\langle x_1,x_2,\dots,x_{2e} \right\rangle$. We define a map $\phi:\mathbb{F}_{p^{2e}}\to \widehat{\mathbb{F}_{p^{2e}}}$ such that $\phi(\sum_{i=1}^{2e}n_ix_i)=\prod_{i=1}^{2e}\left( \chi_{x_i}\right)^{n_i}=\chi_{\sum_{i=1}^{2e}n_ix_i}$ (say), where $\chi_{x_i}$\textquotesingle s are characters for all $ i=1,2,\dots, 2e$ such that $\chi_{x_i}(x_i)=1,\ \chi_{x_i}(x_j)=\xi^{k_{ij}}, \ \forall i<j\leq 2e$ and $\chi_{x_i}(x_j)=\xi^{-k_{ji}}, \ \forall 1\leq j<i$ for each $ i=1,2,\dots, 2e$, where $  \xi  $ is the primitive $ p$-th root of unity and $ 0\leq k_{ij}\leq (p-1) $. \\
		The following system of equations results from continuing along the same path as the Theorem \ref{symduality},
				$$\begin{bmatrix}
			0&-k_{12}&-k_{13}&\dots&-k_{1(2e)}\\
			k_{12}&0&-k_{23}&\dots&-k_{2(2e)}\\
			k_{13}&k_{23}&0&\dots&-k_{3(2e)}\\
			\vdots&\vdots&\vdots&\ddots&\vdots\\
			k_{1(2e)}&k_{2(2e)}&k_{3(2e)}&\dots&0
		\end{bmatrix}  \begin{bmatrix}
			n_1\\n_2\\n_3\\ \vdots\\n_{2e}
		\end{bmatrix}\equiv\begin{bmatrix}
			0\\0\\0\\ \vdots\\0
		\end{bmatrix}\mod p.$$
		The above system of equations has only a trivial solution under mod p if and only if the above skew-symmetric matrix is invertible over $\mathbb{F}_p$. Then, $\phi$ is an isomorphism.  Now, it remains to show that  $\chi_{x}(x)=1,\  \forall x \in \mathbb{F}_{p^{2e}}$. Let $x=\sum_{i=1}^{2e}n_ix_i$, then $ \mathlarger{\mathlarger {\chi}}_ {(\sum_{i=1}^{2e}n_ix_i)}\left(  \sum_{i=1}^{2e}n_ix_i\right)=\prod_{i,j=1}^{2e}\left( \chi_{x_i}(x_j)\right)^{n_in_j}\left(\chi_{x_j}(x_i)\right)^{n_jn_i}=\prod_{i,j=1}^{2e}(\xi^{k_{ij}})^{n_in_j}(\xi^{-k_{ij}})^{n_jn_i}=1$. Hence, we get the required duality. Now, it is  known from \cite{invertibleskew-symmetric} and \cite[Exercise 199]{invertiblesymmetric} that the number of  symmetric matrices in $GL(2e-1,p)$ equals the number of skew-symmetric matrices in  $GL(2e,p)$. Thus, as a result of equation \ref{sym}, we have $ 	p^{e(e-1)}(p-1)(p^3-1)\dots(p^{2e-1}-1) $ dualities over $\mathbb{F}_{p^{2e}}$ such that $\chi_x(x)=1$, for all $x\in\mathbb{F}_{p^{2e}}$.
	\end{proof}
	\begin{remark}
		\begin{enumerate}
			\item In the above theorem, we specifically consider the case of an even number $2e$. This is because a skew-symmetric matrix is invertible if and only if the order of the matrix is even. 
			\item Over $\mathbb{F}_{p^{2e}}, \ p>2$, there are exactly $p^{e(e-1)}(p-1)(p^3-1)\dots(p^{2e-1}-1) $ non-symmetric dualities for which $\chi_x(x)=1, \ \forall x\in \mathbb{F}_{p^{2e}}.$
			\item Over $\mathbb{F}_{2^{2e}}$, if $\chi_x(x)=1$, for all $x\in\mathbb{F}_{2^{2e}}$, then corresponding duality will be symmetric by Lemma \ref{chi-x-y-not1} (since -1 is self inverse). Hence, there are exactly $2^{e(e-1)}(2-1)(2^3-1)\dots(2^{2e-1}-1) $  symmetric dualities such that $\chi_x(x)=1, \ \forall x\in \mathbb{F}_{2^{2e}}.$
		\end{enumerate}
	\end{remark} 
	\begin{theorem}
		There  is no duality over $ \mathbb{F}_{p^{2e+1}},$ such that  $\chi_{x}(x)=1, \ \forall x \in \mathbb{F}_{p^{2e+1}}$.
	\end{theorem}
	\begin{proof}
		There does not exists an invertible skew-symmetric matrix of order $2e+1$. Consequently, over $\mathbb{F}_{p^{2e+1}}$,  there is no duality  such that $\chi_{x}(x)=1,\ \forall x \in \mathbb{F}_{p^{2e+1}}$.
	\end{proof}
	We are now defining a subclass of non-symmetric dualities  where each element is self-orthogonal. These dualities will be known as  skew-symmetric dualities, and this specific class  is represented by $\mathcal{A}$, throughout the paper. \\\\
	\textbf{Note:} For $p=2$, this class $\mathcal{A}$ is  the set of symmetric dualities such that $\chi_x(x)=1$, for all $x\in \mathbb{F}_{2^{2e}}$.\\
	
	\begin{exmp}
		Possible skew-symmetric dualities over the additive group $(\mathbb{F}_{3^2},+)=\left\langle 1,\omega \right\rangle=\{0,1,2,\omega,\omega+1,\omega+2,2\omega,2\omega+1,2\omega+2\}$ are as follows.
		\begin{enumerate}
			\item $ M_1 $ is defined as, $\phi(1)=\chi_1,\ \phi(\omega)=\chi_\omega$, where $\chi_1(1)=1,\ \chi_1(\omega)=\alpha$ and $\chi_\omega(1)=\alpha^2,\ \chi_\omega(\omega)=1$ and  $\alpha$ is 3rd root of unity.
			\item $M_2$ is defined as,  $\phi'(1)=\chi'_1,\ \phi'(\omega)=\chi'_\omega$, where $\chi'_1(1)=1,\ \chi'_1(\omega)=\alpha^2$ and $\chi'_\omega(1)=\alpha,\ \chi'_\omega(\omega)=1$ and $\alpha$ is 3rd root of unity.
		\end{enumerate} 
		We can completely define these dualities by the following character tables:\\
				
					\begin{center}
				\begin{tabular}{|c|c|c|c|c|c|c|c|c|c|}
					\hline
					$ \bm{M_1} $&0&1&2&$ \omega $&$ \omega+1 $&$ \omega+2 $&$ 2\omega $&$ 2\omega+1 $&$ 2\omega+2 $\\
					\hline
					0&1&1&1&1&1&1&1&1&1\\
					\hline
					1&1&1&1&$ \alpha $&$ \alpha $&$ \alpha $&$ \alpha^2 $&$ \alpha^2 $&$ \alpha^2 $\\
					\hline
					2&1&1&1&$ \alpha^2 $&$ \alpha^2 $&$ \alpha^2 $&$ \alpha $&$ \alpha $&$ \alpha $\\
					\hline
					$ \omega $&1&$ \alpha^2 $&$ \alpha $&1&$ \alpha^2 $&$ \alpha $&1&$ \alpha^2 $&$ \alpha $\\
					\hline
					$ \omega+1 $&1&$ \alpha^2 $&$ \alpha $&$ \alpha $&1&$ \alpha^2 $&$ \alpha^2 $&$ \alpha $&1\\
					\hline
					$ \omega+2 $&1&$ \alpha^2 $&$ \alpha $&$ \alpha^2 $&$ \alpha $&1&$ \alpha $&1&$ \alpha^2 $\\
					\hline
					$ 2\omega $&1&$ \alpha $&$ \alpha^2 $&1&$ \alpha $&$ \alpha^2 $&1&$ \alpha $&$ \alpha^2 $\\
					\hline
					$ 2\omega+1 $&1&$ \alpha $&$ \alpha^2 $&$ \alpha $&$ \alpha^2 $&1&$ \alpha^2 $&1&$ \alpha $\\
					\hline
					$ 2\omega+2 $&1&$ \alpha $&$ \alpha^2 $&$ \alpha^2 $&1&$ \alpha $&$ \alpha $&$ \alpha^2 $&1\\
					\hline
				\end{tabular}\vspace{0.2cm}
				\begin{tabular}{|c|c|c|c|c|c|c|c|c|c|}
					\hline
					$\bm{M_2}$&$0$&$1$&$2$&$\omega$&$\omega+1$&$\omega+2$&$2\omega$&$2\omega+1$&$2\omega+2$\\
					\hline
					$0$&$1$  &$1$  &$1$  &$1$  & $1$ & $1 $&$1$  &$1 $ &$1$  \\
					\hline
					$1$& $1$ &$1$ &$1$  &$\alpha^2 $ & $\alpha^2$ &$\alpha^2 $ &$\alpha$  &$\alpha$  &$\alpha$  \\
					\hline
					$2$&$1$  &$1$  &$1$  &$\alpha$  &$\alpha$  &$\alpha$  &$\alpha^2$  &$\alpha^2$  &$\alpha^2$  \\
					\hline
					$\omega$& $1$ &$\alpha$  &$\alpha^2$  &$1$  &$\alpha$  &$\alpha^2$  &$1$  &$\alpha$  &$\alpha^2$  \\
					\hline
					$\omega+1$& $1$ &$\alpha$  &$\alpha^2 $ &$\alpha^2 $ &$1 $ &$\alpha $ &$\alpha$  &$\alpha^2 $ &$1 $ \\
					\hline
					$ \omega+2 $& $ 1 $ & $ \alpha $ & $ \alpha^2 $ &$ \alpha $  &$ \alpha^2 $  &$ 1 $  &$ \alpha^2 $  &$ 1 $  &$ \alpha $  \\
					\hline
					$ 2\omega $& $ 1 $ &$ \alpha^2 $  &$ \alpha $  &$ 1 $  & $ \alpha^2 $ &$ \alpha $  &$ 1 $  &$ \alpha^2 $  &$ \alpha $  \\
					\hline
					$ 2\omega+1 $& $ 1 $ &$ \alpha^2 $  &$ \alpha $  &$ \alpha^2 $  & $ \alpha $ & $ 1 $ & $ \alpha $ &$ 1 $  &$ \alpha^2 $  \\
					\hline
					$ 	2\omega+2 $&$ 1 $  &$ \alpha^2 $  &$ \alpha $  &$ \alpha $  &$ 1 $  &$ \alpha^2 $  &$ \alpha^2 $  &$ \alpha $  & $ 1 $\\
					\hline
				\end{tabular}
			\end{center}
			
			\end{exmp}
	\section{ACD codes over skew-symmetric dualities  }\label{ACD code on skew-symmetric dualities}
	Let $C$ be an additive code over  $\mathbb{F}_q$ and $\mathcal{G}$ be an $m\times n$ generator matrix of $C$ with entries over $\mathbb{F}_q$. Let $\mathcal{G}_i$ be the $i$-th row of the matrix $\mathcal{G}$. Then  $ij$-th entry of the matrix $ \mathcal{G}\odot_M\mathcal{G}^T $ is defined as $(\mathcal{G}\odot_M\mathcal{G}^T)_{ij}=\chi_{\mathcal{G}_i}(\mathcal{G}_j)=\xi^{k_{ij}}$, where $\xi$ is the primitive $p$-th root of unity and $0\leq k_{ij}\leq p-1$. Define the matrix $\log_{\xi}(\mathcal{G}\odot_M\mathcal{G}^T)$ as follows:
	\begin{align*}
		(\log_{\xi}(\mathcal{G}\odot_M\mathcal{G}^T))_{ij}=&\log_{\xi}(\mathcal{G}\odot_M\mathcal{G}^T)_{ij}=\log_{\xi}(\xi^{k_{ij}})\equiv k_{ij}\mod p, 
	\end{align*}
	where $\log_\xi$ represents the principal logarithm of complex numbers to the base $\xi$. 
	Here, we shall provide a necessary and sufficient condition that must be met for an additive code $C$ to be an ACD code concerning  arbitrary dualities.
	\begin{theorem}\label{iff condition on ACD code}
		Let $\mathcal{G}$ be an  $m\times n$ generator matrix of an additive  code $C$ over $\mathbb{F}_{p^e}$. Then $C$ is an ACD code  with respect to the duality $M$ if and only if the matrix $\log_{\xi}(\mathcal{G}\odot_M\mathcal{G}^T)$ is invertible over $\mathbb{F}_{p}$ , where $\xi$ is the primitive $p$-th root of unity.
	\end{theorem}
	\begin{proof}	
	Let $\textbf{v}\in C\cap C^M$, then $\textbf{v}=\sum_{i=1}^{m}n_i\mathcal{G}_i$ and $\chi_{\textbf{v}}(\mathcal{G}_j)=1,\ \forall 1\leq j \leq m.$ We have $\chi_{(\sum_{i=1}^{m}n_i\mathcal{G}_i)}(\mathcal{G}_j)=\prod_{i=1}^{m}(\chi_{\mathcal{G}_i}(\mathcal{G}_j))^{n_i}=\prod_{i=1}^{m}\xi^{k_{ij}n_i}=\xi^{\sum_{i=1}^{m}{k_{ij}n_i}}=1$, where $\xi$ is the primitive $p$-th root of unity and $0\leq k_{ij}\leq p-1$. It gives that  $ \sum_{i=1}^{m}{k_{ij}n_i}\equiv 0\mod p, \ \forall 1\leq j \leq m$. Therefore, we have the following system of equations 
		
		\begin{center}
			$\begin{bmatrix}
				k_{11}&k_{21}&k_{31}&\dots&k_{m1}\\
				k_{12}&k_{22}&k_{32}&\dots&k_{m2}\\
				k_{13}&k_{23}&k_{33}&\dots&k_{m3}\\
				\vdots&\vdots&\vdots&\ddots&\vdots\\
				k_{1m}&k_{2m}&k_{3m}&\dots&k_{mm}
			\end{bmatrix}  \begin{bmatrix}
				n_1\\n_2\\n_3\\ \vdots\\n_{m}
			\end{bmatrix}\equiv\begin{bmatrix}
				0\\0\\0\\ \vdots\\0
			\end{bmatrix} \mod p.$
		\end{center}
		In particular, we have
		$$(\log_{\xi}(\mathcal{G}\odot_M\mathcal{G}^T))^T\begin{bmatrix}
			n_1\\n_2\\n_3\\ \vdots\\n_{m}
		\end{bmatrix}\equiv\begin{bmatrix}
			0\\0\\0\\ \vdots\\0
		\end{bmatrix} \mod p.$$
		The above system of equations has only a trivial solution if and only if the matrix $ \log_{\xi}(\mathcal{G}\odot_M\mathcal{G}^T)$ is invertible over $\mathbb{F}_p$. This argument completes the proof.
	\end{proof}
	\begin{exmp}
		Let $C$ be an additive code over $\mathbb{F}_{3^2}$ generated by a matrix $\mathcal{G}=\begin{bmatrix}
			1&\omega\\
			\omega+1&0
		\end{bmatrix}$.  We have  $\mathcal{G}\odot_{M_1}\mathcal{G}^T=\begin{bmatrix}
			1&\alpha\\
			\alpha^2&1
		\end{bmatrix}$, where $\alpha$ is the 3-rd root of unity and $\log_{\alpha}(\mathcal{G}\odot_{M_1}\mathcal{G}^T)=\begin{bmatrix}
			0&1\\
			-1&0
		\end{bmatrix}$ is  an invertible matrix over $\mathbb{F}_3$. Then $C=\{(0,0),(1,\omega),(\omega+1,0),(\omega+2,\omega),(2\omega,\omega),(\omega,2\omega),(2\omega+1,2\omega),(2,2\omega),(2\omega+2,0)\}$.  The dual of the code $C$ with respect to duality $M_1$ is $C^{M_1}=\{(0,0),(0,\omega),(0,2\omega),(\omega+1,1),(\omega+1,\omega+1),(\omega+1,2\omega+1),(2\omega+2,2),(2\omega+2,\omega+2),(2\omega+2,2\omega+2)\}$. It is clear that $C\cap C^{M_1}=\{(0,0)\}$. Hence, $C$ is an ACD code for duality $M_1$. 
	\end{exmp}
	\begin{exmp}
		Let $C$ be an additive code over $\mathbb{F}_{3^2}$ generated by a matrix $\mathcal{G}=\begin{bmatrix}
			1&\omega\\
			\omega &1
		\end{bmatrix}$. Then $C=\{(0,0),(1,\omega),(\omega,1),(2,2\omega),(2\omega,2),(1+\omega,1+\omega),(1+2\omega,\omega+2),(2+\omega,2\omega+1),(2+2\omega,2+2\omega)\}$. It is easy to see that $C=C^{M_1}.$ Hence, it is not an ACD code. Also, we observe that $\log_{\alpha}(\mathcal{G}\odot_{M_1}\mathcal{G}^T)=\begin{bmatrix}
			0&0\\
			0&0
		\end{bmatrix}$ is a non-invertible matrix.
	\end{exmp}
	
	The following theorem states that if the number of generators of an additive code  $C$  over $\mathbb{F}_{p^{2e}}$ is odd, it cannot satisfy the properties required for being an ACD code under any duality of the class $\mathcal{A}$. 
	\begin{theorem}
		For any duality $M\in \mathcal{A}$, if $C$ is an ACD code with the generator matrix $\mathcal{G}$ of order $m\times n$ over $\mathbb{F}_{p^{2e}}, $ then  $m$ is even.
	\end{theorem}
	\begin{proof}
		Let $\mathcal{G}_i$ be the $i$-th row of the matrix $\mathcal{G}$.	For any duality $M\in\mathcal{A}$, we have  $\chi_{\mathcal{G}_i}(\mathcal{G}_j)=\begin{cases}
			1,& \text{\ if } \ i=j,\\
			\xi^{k_{ij}}, & \text{ if } \ i< j,  \\
			\xi^{-k_{ji}}, &\text{ if } \ i>j  
		\end{cases}. $ Consequently, one can deduce that the matrix $ \log_{\xi}(\mathcal{G}\odot_M\mathcal{G}^T) $ is skew-symmetric of order $m\times m$. Hence, by the Theorem \ref{iff condition on ACD code} and the fact that odd order skew-symmetric matrices are not invertible, $m$ is even. 
	\end{proof}
	The next two corollaries provide construction techniques for generating ACD codes  over skew-symmetric dualities. These insights offer valuable tools for constructing ACD codes in this particular context.
	\begin{corollary}\label{ACD1}
		Let $C$ be an additive code over $\mathbb{F}_{p^{2e}}$ with generator matrix $\mathcal{G}$ of order $2s\times n$. Let $\mathcal{G}_i$ be the $i$-th row of the generator matrix $\mathcal{G}$ such that $\chi_{\mathcal{G}_i}(\mathcal{G}_j)=\xi$, for all $1\leq i<j\leq 2s\ \text{and}\  M\in \mathcal{A}$, where $\xi$ is the primitive $p$-th root of unity. Then $C$ is an ACD code.  
	\end{corollary}
	\begin{proof}
		Let $\textbf{v}\in C\cap C^M$, then $\textbf{v}=\sum_{i=1}^{2s}n_i\mathcal{G}_i$ and $\chi_{\sum_{i=1}^{2s}n_i\mathcal{G}_i}(\mathcal{G}_j)=1, \ \forall 1\leq j \leq 2s.$ Thus, $\chi_{\sum_{i=1}^{2s}n_i\mathcal{G}_i}(\mathcal{G}_j)=\prod_{i=1}^{2s}\chi_{\mathcal{G}_i}(\mathcal{G}_j)^{n_i}=\prod_{i=1}^{j-1}\xi^{n_i}\prod_{i=j+1}^{2s}\xi^{-n_i}=\xi^{\sum_{i=1}^{j-1}n_i-\sum_{i=j+1}^{2s}n_i}=1$. This implies that $\sum_{i=1}^{j-1}n_i-\sum_{i=j+1}^{2s}n_i\equiv 0\mod p, \ \forall 1\leq j \leq 2s.$ This system of equations has only a trivial solution. Hence, the code generated by $\mathcal{G}$ is an ACD code.
	\end{proof}
	\begin{exmp}
		Let $\mathcal{G}=\begin{bmatrix}
			1 & \omega &0\\
			1&1&\omega\\
			0&1&1\\
			2\omega&\omega& \omega+1
		\end{bmatrix}$ be a generator matrix of a code over $\mathbb{F}_{3^2}$. Consider the dualities $M_1$ and $M_2$. Since $\chi_{\mathcal{G}_i}(\mathcal{G}_j)_{M_1}=\alpha^2$, for all $1\leq i<j\leq 4$, and $\chi_{\mathcal{G}_i}(\mathcal{G}_j)_{M_2}=\alpha$, for all $1\leq i<j\leq 4$, where $\alpha$ is the primitive $ 3 $rd root of unity. Therefore, matrix $\mathcal{G}$ satisfies the condition of Corollary \ref{ACD1}. Now, one can deduce that  $\log_{\alpha}(\mathcal{G}\odot_{M_1}\mathcal{G}^T)=\begin{bmatrix}
			0&2&2&2\\-2&0&2&2\\-2&-2&0&2\\-2&-2&-2&0
		\end{bmatrix}$ and $\log_{\alpha}(\mathcal{G}\odot_{M_2}\mathcal{G}^T)=\begin{bmatrix}
			0&1&1&1\\-1&0&1&1\\-1&-1&0&1\\-1&-1&-1&0
		\end{bmatrix}$. These matrices are invertible over $\mathbb{F}_{3}$. In view of the Theorem \ref{iff condition on ACD code}, the code generated by $\mathcal{G}$ is ACD code for both dualities $M_1$ and $M_2$.	
	\end{exmp}
	
	\begin{corollary}\label{ACD2}
		Let $C$ be an additive code over $\mathbb{F}_{p^{2e}}$ with generator matrix $\mathcal{G}$ of order $2s\times n$. Let $\mathcal{G}_i$ be the $i$-th row of the matrix $ \mathcal{G} $ and $S=\{(1,2),(3,4),(5,6),\dots,(2s-1,2s)\}$ such that $\chi_{\mathcal{G}_i}(\mathcal{G}_j)=\begin{cases}
			\xi_i & (i,j)\in S,\\
			\xi^{-1}_j & (j,i)\in S,\\
			1 & O/w
		\end{cases}$,  where $\xi_i$\textquotesingle s are any primitive $p$-th roots of unity for all $i=1,3,5,\dots,2s-1$ and $M\in\mathcal{A}$. Then $C$ is an ACD code.
	\end{corollary}
	
	\begin{proof}
		Let $\textbf{v}\in C\cap C^M$, then $\textbf{v}=\sum_{i=1}^{2s}n_i\mathcal{G}_i$ and $\chi_{\sum_{i=1}^{2s}n_i\mathcal{G}_i}(\mathcal{G}_j)=1, \ \forall 1\leq j \leq 2s.$ Thus, $\chi_{\sum_{i=1}^{2s}n_i\mathcal{G}_i}(\mathcal{G}_j)=\prod_{i=1}^{2s}\chi_{\mathcal{G}_i}(\mathcal{G}_j)^{n_i}=\begin{cases}
			\xi_{j-1}^{n_{j-1}} & j\ \text{is even},\\
			\xi_j^{-n_{j+1}} & j \ \text{is odd}
		\end{cases}=1$. This implies $n_{j-1}\equiv 0\mod p$, when $j$ is even, and $-n_{j+1}\equiv 0\mod p$, when $j$ is odd. Hence, $C\cap C^M=\{0\}.$
	\end{proof}
\noindent	\textbf{Note:} For a clear view of Corollary \ref{ACD2}, one can deduce that if $\mathcal{G}$ is a generator matrix of an additive code $C$ such that $$\mathcal{G}\odot_{M}\mathcal{G}^T=\begin{bmatrix}
		1&\xi_1&1&1&\dots&1&1\\
		\xi_{1}^{-1}&1&1&1&\dots&1&1\\
		1&1&1&\xi_3&\dots&1&1\\
		1&1&\xi_{3}^{-1}&1&\dots&1&1\\
		\vdots&\vdots&\vdots&\vdots&\ddots&\vdots&\vdots&\\
		1&1&1&1&\dots&1&\xi_{2s-1}\\
		1&1&1&1&\dots&\xi_{2s-1}^{-1}&1\\
	\end{bmatrix}$$
	Then $C$ is an ACD code.
	\begin{exmp}
		Let $\mathcal{G}=\begin{bmatrix}
			1&\omega&2&0&0\\
			\omega&\omega&1&0&0\\
			0&0&0&1&\omega\\
			0&0&0&2\omega&\omega
		\end{bmatrix}$ be a generator matrix for an additive code $C$ over $ \mathbb{F}_{3^2}$. Consider the duality $M_1$ and $M_2$. We have $\mathcal{G}\odot_{M_1}\mathcal{G}^T=\begin{bmatrix}
			1&\alpha&1&1\\
			\alpha^2&1&1&1\\
			1&1&1&\alpha^2\\
			1&1&\alpha&1
		\end{bmatrix}$ and $\mathcal{G}\odot_{M_2}\mathcal{G}^T=\begin{bmatrix}
			1&\alpha^2&1&1\\
			\alpha&1&1&1\\
			1&1&1&\alpha\\
			1&1&\alpha^2&1
		\end{bmatrix}$, where $\alpha$ is the primitive $ 3 $rd root of unity. Hence, it satisfy the conditions of Corollary \ref{ACD2}. Now, one can deduce that $\log_{\alpha}(\mathcal{G}\odot_{M_1}\mathcal{G}^T)=\begin{bmatrix}
			0&1&0&0\\
			2&0&0&0\\
			0&0&0&2\\
			0&0&1&0
		\end{bmatrix}$ and $\log_{\alpha}(\mathcal{G}\odot_{M_2}\mathcal{G}^T)=\begin{bmatrix}
			0&2&0&0\\
			1&0&0&0\\
			0&0&0&1\\
			0&0&2&0
		\end{bmatrix}$. These matrices are invertible over $\mathbb{F}_3$. Hence, by  Theorem \ref{iff condition on ACD code}, the code is ACD code for both dualities $M_1$ and $M_2$.
	\end{exmp}

	\section{Bounds of ACD codes on skew-symmetric dualities }\label{bounds on skew-symmetric dualities}
	\begin{defn}
		Let $ [n,p^k,d]_A $ denote the additive code of length $ n $ with cardinality $p^k$ and  distance $ d $ over the finite field $ \mathbb{F}_{p^e}$.  We define the highest possible distance of   ACD codes over $\mathbb{F}_{p^e}$ with respect to the duality $M$ as $ACD_{{p^e},M}[n,k]=\max\{d\mid\ [n,p^k,d]_A  \ \text{ACD code exists}\}$.
	\end{defn}
	The Singleton bound for an $ [n,p^k,d]_A$  additive code over $\mathbb{F}_{p^e}$ is $d\leq n-\lceil\frac{k}{e}\rceil+1.$ If a code attains the Singleton bound, then it is known as Maximum Distance Separable (MDS) code.
	We will say the two vectors $\textbf{u},\textbf{v}\in \mathbb{F}_{p^e}^n$ are $p$-linearly independent  in the sense that they are linearly independent over $\mathbb{F}_p$.
	\begin{lemma}\label{chi(x)(y) eq xi}
		Let $ \mathbb{F}_{p^{2e}},\ e\geq 1 $ be a finite field, then there exist $x,y\in\mathbb{F}_{p^{2e}}$ such that $\chi_{x}(y)=\xi\neq1$  for all dualities $M\in \mathcal{A}$, where $\xi$ is any primitive $p$-th root of unity.
	\end{lemma}
	\begin{proof} Let the generators of $\mathbb{F}_{p^{2e}},\ e\geq1 $ be $x_1,x_2,\cdots,x_{2e}$ such that  $\mathbb{F}_{p^{2e}}=\left\langle x_1,x_2,\cdots,x_{2e}\right\rangle $. If, for some $ x_i$ and $x_j $, $ \chi_{x_i}(x_j)=\xi\neq1$, for $i\neq j$, then we are done. If $\chi_{x_i}(x_j)=1$, for all $1\leq i,j\leq 2e$, then $\chi_{x_i}(x)=1,\ \forall x\in \mathbb{F}_{p^{2e}}$. It follows that $x_i=0$. This gives a contradiction.
	\end{proof}
	\begin{lemma}\label{chi x y eq to 1}
		Let $\mathbb{F}_{p^{2e}},\ e>1 $ be a finite field, then there exist $x,y\in \mathbb{F}_{p^{2e}}$ such that $x$ and  $y$ are $p$-linearly independent and $\chi_x(y)=1$, for all dualities $M\in \mathcal{A}.$
	\end{lemma}
	\begin{proof}
		Let the generators of $\mathbb{F}_{p^{2e}},\ e>1 $ be $x_1,x_2,\cdots,x_{2e}$ such that  $\mathbb{F}_{p^{2e}}=\left\langle x_1,x_2,\cdots,x_{2e}\right\rangle $. If, for some $ x_i$ and $x_j $, $ \chi_{x_i}(x_j)=1$, for $i\neq j$, then we are done. Let $\xi$ be the primitive $p$-th root of unity. Let $ j\neq k\neq l $ and  $\chi_{x_j}{(x_l)}=\xi^a$, $\chi_{x_k}{(x_l)}=\xi^b$, for some $0<a,b<p$, then $\chi_{(-ba^{-1})x_j+x_k}(x_l)=\chi_{x_j}(x_l)^{(-ba^{-1})}\chi_{x_k}(x_l)=\xi^{a(-ba^{-1})}\xi^b=1$.  
	\end{proof}
	\begin{lemma}\label{chi x y can not be 1 for e eq 1}
		Let $\mathbb{F}_{p^{2}} $ be a finite field, then there does not exist $x,y\in \mathbb{F}_{p^2}$ such that $x$ and  $y$ are $p$-linearly independent and $\chi_x(y)=1$, for any duality $M\in \mathcal{A}.$
	\end{lemma}
	\begin{proof}
		Let the generators of $\mathbb{F}_{p^{2}}$ be $x_1$ and $x_2$ such that  $\mathbb{F}_{p^{2}}=\left\langle x_1,x_2\right\rangle $. For any duality $M\in \mathcal{A}$, $\chi_{x_1}(x_2)=\xi\neq 1$, where $\xi$ is the primitive  $p$-th root of unity. If  $  \mathlarger{\chi}_{(l_1x_1+l_2x_2)}(m_1x_1+m_2x_2)=1 $, for some $ 0<l_1,l_2,m_1,m_2<p$, then we have $l_1m_2-l_2m_1\equiv 0 \mod p$. This implies  that  $ l_1x_1+l_2x_2$ and $m_1x_1+m_2x_2 $ are $p$-linearly dependent.
	\end{proof} 
	In the next theorem, we shall provide the lower bound on the highest possible minimum distance of ACD codes concerning the class $\mathcal{A}$ over $\mathbb{F}_{p^{2e}}$.
	\begin{theorem}\label{ACDbound}
		Let $\mathbb{F}_{p^{2e}},\  p\neq 2$ be a finite field, and $s,\ n\in \mathbb{N}$. If $s\leq n$, then  $ACD_{p^{2e},M}[n,2s]\geq\lfloor\frac{n}{s}\rfloor$, for all dualities $M\in \mathcal{A}$, where $ \lfloor \cdot \rfloor $ denotes the largest integer less than or equal to $(\cdot)$.
	\end{theorem}
	\begin{proof}
		Let $C$ be an additive code generated by a matrix $\mathcal{G}$, where
		\setlength{\arraycolsep}{2pt}
		\renewcommand{\arraystretch}{0.9}
		$\mathcal{G}=\begin{bmatrix}
			B_1& & &&\\
			&B_2&&\bigzero &\\
			&&\ddots&&\\
			&\mbox{\Huge 0}&&B_{s-1}&\\
			&&&&B_s
		\end{bmatrix}_{2s\times n}$,  and $B_i$\textquotesingle s are block matrices of appropriate size. Now, by Lemma \ref{chi(x)(y) eq xi}, we have $x$ and $y\in \mathbb{F}_{p^{2e}}$ such that $\chi_{x}(y)=\xi\neq 1$, where $\xi$ is any primitive $p$-th root of unity.\\\\
		\textbf{Case 1:}	If $\lfloor\frac{n}{s}\rfloor\neq mp$, then choose $B_i=\begin{pmatrix}
			x&x&\cdots&x\\
			y & y &\cdots &y
		\end{pmatrix}_{2\times\lfloor\frac{n}{s}\rfloor}$ for $i=1,2,\cdots,s-1$, and  $B_s=\begin{pmatrix}
			x&x&\cdots&x&0&\cdots&0\\
			y & y &\cdots &y&0&\cdots&0
		\end{pmatrix}_{2\times\left( \lfloor\frac{n}{s}\rfloor+n\mod s\right)}$. Since $\left( \chi_x(y)\right)^{\lfloor\frac{n}{s}\rfloor}=\xi^{\lfloor\frac{n}{s}\rfloor}\neq 1$, therefore, by Corollary\ref{ACD2}, $C$ is an ACD code with $d=\lfloor\frac{n}{s}\rfloor.$\\\\
		\textbf{Case 2:} If $\lfloor\frac{n}{s}\rfloor= mp$, then choose $B_i=\begin{pmatrix}
			x&x&\cdots&x&y\\
			y & y &\cdots &y& x
		\end{pmatrix}_{2\times\lfloor\frac{n}{s}\rfloor}$ for $i=1,2,\cdots,s-1$, and  $B_s=\begin{pmatrix}
			x&x&\cdots&x&y&0&\cdots&0\\
			y & y &\cdots &y&x &0&\cdots&0
		\end{pmatrix}_{2\times\left( \lfloor\frac{n}{s}\rfloor+n\mod s\right)}$.  Since $\left( \chi_x(y)\right)^{\lfloor\frac{n}{s}\rfloor-1} \chi_y(x)=\left(\chi_x(y)\right)^{\lfloor\frac{n}{s}\rfloor-2}=\xi^{\lfloor\frac{n}{s}\rfloor-2}\neq 1$, therefore, by Corollary\ref{ACD2}, $C$ is an ACD code with $d=\lfloor\frac{n}{s}\rfloor.$\\ 
	\end{proof}

	\begin{theorem}\label{ACD[n,2n-2] eq 2}
		For all dualities  $M\in \mathcal{A}$, we have $ACD_{p^2,M}[n,2n-2]=2$. 
	\end{theorem}
	\begin{proof}
		
		From the Singleton bound, we have $d\leq 2$. Let $\mathbb{F}_{p^2}=\left\langle x_1,x_2 \right\rangle$. If $n\neq mp$ for all  $m\in\mathbb{N}$, then take  $C^M=\langle(x_1,x_1,x_1,\dots,x_1),(x_2,x_2,x_2,\dots, x_2)\rangle$, otherwise   choose $C^M=\langle(x_1,x_1,x_1,\dots,x_1,x_2),(x_2,x_2,x_2,\dots,x_2, x_1)\rangle$. In both  cases, $C^M$ is an ACD code with distance $n$ for $M\in \mathcal{A}$ .  From the Theorem \ref{C-M-M- equl to C}, we have $(C^M)^M=C$.  Thus, it follows that  $C$ is an ACD code with parameter $[n,p^{2n-2}]$. If $C$ contains a vector of weight $1$ say $\bm{y}=(0,\dots,0,y_i,0,\dots,0)$, then $\chi_{y_i}(x_1)=1$ and $\chi_{y_i}(x_2)=1$. This is not possible because a non-zero element can not be orthogonal to the whole space. Hence, the result holds.
	\end{proof}

	\begin{theorem}\label{ACD bound for p eq 2}
		Let $q=2^{2e}, \ e> 1$, and $s, n\in \mathbb{N}$. If $s\leq n$ then $ACD_{q,M}[n,2s]\geq\lfloor\frac{n}{s}\rfloor$, for all dualities $M\in \mathcal{A}.$
	\end{theorem}
	\begin{proof}
		Let $C$ be an additive code generated by a matrix $\mathcal{G}$, where
		\setlength{\arraycolsep}{2pt}
		\renewcommand{\arraystretch}{0.9}
		$\mathcal{G}=\begin{bmatrix}
			B_1& & &&\\
			&B_2&&\bigzero &\\
			&&\ddots&&\\
			&\mbox{\Huge 0}&&B_{s-1}&\\
			&&&&B_s
		\end{bmatrix}_{2s\times n}$, $B_i$\textquotesingle s are block matrices of appropriate size. Now, by Lemma \ref{chi(x)(y) eq xi}, we have $x$ and $y\in \mathbb{F}_{2^{2e}}$ such that $\chi_{x}(y)=-1$, and by Lemma \ref{chi x y eq to 1}, we have $a$ and $b\in \mathbb{F}_{2^{2e}}$ such that $\chi_{a}(b)=1$.\\
		\textbf{Case 1:} If $\lfloor\frac{n}{s}\rfloor$ is an odd integer, then the  construction of an ACD code is the same as Case 1 of Theorem \ref{ACDbound}.  \\
		\textbf{Case 2:} If $\lfloor\frac{n}{s}\rfloor$ is an even integer, then choose $B_i=\begin{pmatrix}
			x&x&\cdots&x&a\\
			y & y &\cdots &y&b
		\end{pmatrix}_{2\times \lfloor\frac{n}{s}\rfloor}$ for $i=1,2,\cdots,s-1$, and  $B_s=\begin{pmatrix}
			x&x&\cdots&x&a&0&\cdots&0\\
			y & y &\cdots &y&b&0&\cdots&0
		\end{pmatrix}_{2\times\left( \lfloor\frac{n}{s}\rfloor+n\mod s\right)}$. Since $\chi_x(y)^{\lfloor\frac{n}{s}\rfloor-1}\chi_a(b)=-1$, therefore,  by Corollary \ref{ACD2}, $C$ is an  ACD code with $d=\lfloor\frac{n}{s}\rfloor.$ 
	\end{proof}
	\begin{corollary}\label{ACDBoundfor2}
		 For all dualities $M\in \mathcal{A}$ and $p\neq 2$, we have $ACD_{p^{2e},M}[n,2]=n$. Moreover, if $e>1$ then $ACD_{2^{2e},M}[n,2]=n$.
	\end{corollary}
	\begin{proof}
		The proof follows from  Theorems \ref{ACDbound} and \ref{ACD bound for p eq 2} for $ s=1.$
	\end{proof}

	\begin{theorem}
	 For all dualities $M\in\mathcal{A}$, we have	$ACD_{4,M}[n,2]$=$\begin{cases}
			n & \text{if $ n $ is an odd integer},\\
			n-1 & \text{if $ n $ is an even integer}
		\end{cases}$.
	\end{theorem}
	\begin{proof}
		If $n$ is an odd integer then the proof is similar to the Case 1 of Theorem \ref{ACDbound} for $s=1$. Let $n$ be an even integer. If possible, let  $ACD_{4,M}[n,2]=n,$ for all dualities $M\in\mathcal{A}$. Then there exists a code $C=\left\langle (x_1,x_2,\cdots,x_n),(y_1,y_2,\cdots,y_n) \right\rangle$ such that $x_i\neq 0,\ y_i\neq 0$ and $x_i\neq y_i$ for all $i$. We observe that there is exactly  one duality over $ \mathbb{F}_4,$ such that  $\chi_x(x)=1$, for all $x\in \mathbb{F}_4$. Note that for such duality we have $\chi_x(y)=-1$, for all $ x\neq y\in \mathbb{F}_4$. Now we have $\prod_{i=1}^{n}\chi_{x_i}(y_i)=(-1)^n=1$, this implies that $C$ can not be an ACD code. Hence, $ACD_{4,M}[n,2]<n$. Now, it is easy to see that  $C=\left\langle (x,x,\cdots,x,0),(y,y,\cdots,y,0)\right\rangle $ is an ACD code for $x\neq y$. Hence, $ACD_{4,M}[n,2]=n-1$.
	\end{proof}
	We conclude this section by summarising the above results in Table \ref{Table 1} for skew-symmetric dualities over $\mathbb{F}_{3^2}$.
	\begin{table*}[h]
		\caption{ACD bounds  for skew-symmetric dualiies over $\mathbb{F}_{3^2}$}
		\begin{center}
			\begin{tabular}{|c|c|c|c|}
				\hline
				& $ M_1 $ &  $ M_2 $& Reason \\
				\hline
				$ ACD_{9,M}[1,2] $	& $ 1 $ & $ 1 $  & By Corollary \ref{ACDBoundfor2} \\
				\hline
				$ ACD_{9,M}[2,2] $	& 2 & 2 & By Corollary \ref{ACDBoundfor2} \\
				\hline
				$ ACD_{9,M}[2,4] $	& 1 & 1 &  The code with these parameters will span the entire space $ \mathbb{F}_{3^2}^2 $. \\
				\hline
				$ ACD_{9,M}[3,2] $	& 3 & 3 & By Corollary \ref{ACDBoundfor2} \\
				\hline
				$ ACD_{9,M}[3,4] $	& 2 & 2 &\makecell*{ By Theorem \ref{ACD[n,2n-2] eq 2}, we have a code generated by $\mathcal{G}= \begin{pmatrix}
					1&1&0\\
					\omega&\omega&0\\
					1&0&1\\
					\omega&0&\omega
				\end{pmatrix} $ \\ is an MDS ACD code with $ d=2$.}\\
				\hline
				$ ACD_{9,M}[3,6] $	& 1 &  1&  The code with these parameters will span the entire space $ \mathbb{F}_{3^2}^3$. \\
				\hline
				$ ACD_{9,M}[4,2] $	& 4 &  4& By Corollary \ref{ACDBoundfor2} \\
				\hline
				$ ACD_{9,M}[4,4] $	&  3 & 3 &\makecell*{ The code generated by $\mathcal{G}= \begin{pmatrix}
					1&0&1&\omega+1\\
					\omega&0&\omega&\omega\\
					0&1&1&\omega\\
					0&\omega&\omega&1
				\end{pmatrix} $\\ is an MDS ACD code with $ d=3 $.}\\
				\hline
				$ ACD_{9,M}[4,6] $	& 2  &  2 & \makecell*{By Theorem \ref{ACD[n,2n-2] eq 2}, we have a code generated by $\mathcal{G}=\begin{pmatrix}
					1&1&1&0\\
					\omega&\omega&\omega&0\\
					1&1&0&1\\
					\omega&\omega&0&\omega\\
					1&0&1&1\\
					\omega&0&\omega&\omega
				\end{pmatrix} $ \\ is an MDS ACD code with $ d=2$.}  \\
				\hline
				$ ACD_{9,M}[4,8] $	& 1 &  1& The code with these parameters will span the entire space $ \mathbb{F}_{3^2}^4$. \\
				\hline
			\end{tabular} 
		\end{center}
		\label{Table 1}
	\end{table*}

	\section{Non-symmetric dualities over $\mathbb{F}_4$}\label{non-symmetric dualities} 
	In \cite{ACDgroupcharacter}, authors have given the bounds on the highest  minimum distance of ACD codes for symmetric dualities over $\mathbb{F}_4$. We find some new quaternary ACD codes over non-symmetric dualities  with better parameter  than the symmetric ones. This motivates us to study ACD codes over non-symmetric dualities. There are exactly two non-symmetric dualities, $ D_1 $ and $ D_2 $, over $\mathbb{F}_4= \left\langle 1,\nu \right\rangle $.
			\begin{center}
			\begin{tabular}{|c|c|c|c|c|}
				\hline
				$ D_1 $& 0 &1  &$ \nu $  &$ \nu+1 $  \\
				\hline
				0&1  & 1 & 1 & 1 \\
				\hline
				1& 1 & -1 & -1 & 1 \\
				\hline
				$ \nu$&  1 & 1 &-1  &  -1 \\
				\hline
				$ \nu+1 $& 1 &  -1&  1& -1 \\
				\hline
			\end{tabular}\hspace{0.5cm}
			\begin{tabular}{|c|c|c|c|c|}
				\hline
				$ D_2 $& 0 &1  &$ \nu  $ &$ \nu+1 $  \\
				\hline
				0& 1 & 1 &  1& 1 \\
				\hline
				1&1  & -1 & 1 & -1 \\
				\hline
				$ \nu $& 1 & -1&-1  & 1 \\
				\hline
				$\nu+1 $& 1 & 1 & -1 & -1 \\
				\hline
			\end{tabular}
		\end{center}
	
		\begin{theorem}\label{ACD_4[n,1]=2 for D1 and D2}
		If $M\in \{D_1,D_2\}$, then $ACD_{4,M}[n,1]$=$\begin{cases}
			n & \text{if $ n $ is an odd integer},\\
			n-1 & \text{if $ n $ is an even integer}
		\end{cases}$.
	\end{theorem}
	\begin{proof}
		We observe that $\chi_x(x)=-1,\ \forall x\in \mathbb{F}^*_4$ for non-symmetric dualities $D_1$ and $D_2$.  Let $\textbf{u}=(u_1,u_2,\dots,u_n)\in\mathbb{F}^n_4$ with $u_i\neq 0$, for all $1\leq i\leq n$ then $\chi_{\bm{u}}(\bm{u})=\prod_{i=1}^{n}\chi_{u_i}(u_i)=(-1)^n$. If $n$ is an odd integer, then the vector $\textbf{u}$ generates an ACD codes with distance $n$.  If  $n$ is an even integer, then  any vector of weight $n$ can not generate an ACD code. Hence, $\bm{u}=(u_1,u_2,\dots,u_{n-1},0)$ generates an ACD code when $n$ is even. This completes the proof. 
	\end{proof}
	\begin{theorem}\label{ACD_4[n,2]=2 for D1 and D2}
		If $M\in \{D_1,D_2\}$, then $ACD_{4,M}[n,2]=n$.
	\end{theorem}
	\begin{proof}
		If $n$ is an odd integer then choose an additive code $C=\left\langle (1,1,\dots,1),(\nu,\nu,\dots,\nu) \right\rangle $. Otherwise choose $C=\left\langle (1,1,\dots,1,\nu),(\nu,\nu,\dots,\nu,1)\right\rangle $.  Clearly, in the both cases, $C$ is an ACD code with $d=n$. 
	\end{proof}
	\begin{theorem}\label{ACD4[n,2n-2]>=2 for D1 and D2}
		If $M\in \{D_1,D_2\}$, then $ACD_{4,M}[n,2n-2]=2$.
	\end{theorem}
	\begin{proof}
		If $n$ is an odd integer, then choose $C^{D_1}=\left\langle (1,1,\dots,1),(\nu,\nu,\dots,\nu) \right\rangle$. Otherwise choose $C^{D_1}=\left\langle (1,1,\dots,1,\nu),(\nu,\nu,\dots,\nu,1) \right\rangle$.  We have $C=(C^{D_1})^{D_2}$. Let $C$ contains a vector $\textbf{x}$ of weight 1, without loss of generality, $\bm{x}=(x_1,0,0,\dots,0)$. Then $\chi_{x_1}(1)=1$ and $\chi_{x_1}(\nu)=1$, which is not possible in $D_2$. Thus, $C$ is an ACD code  for duality $D_1$ with $d\geq 2.$  Similarly, it can be proved for $D_2$. From the  Singleton bound for the additive code $[n,2^{(2n-2)}]$, we have $d\leq 2$. Hence, combining the results, we get $ACD_{4,M}[n,2n-2]=2$. 
	\end{proof}
	\begin{theorem}\label{ACD4[n,2n-1]=1 for D1 and D2}
		If $M\in \{D_1,D_2\}$ then $ACD_{4,M}[n,2n-1]=1$.
	\end{theorem}
	\begin{proof}
		According to the Singleton bound, we have $ACD_{4,M}[n,2n-1]\leq 1$. If $C$ is an $[n,2^1]$ ACD  code,  then 	 $C^{M^T}$ is an $[n,2^{2n-1}]$ ACD code  with respect to duality $M$ (as stated in Proposition \ref{C^{M^T} is ACD}). Consequently, we conclude that 	$ACD_{4,M}[n,2n-1]=1$, for all dualities $M$ over $\mathbb{F}_{4}$. 
	 
	\end{proof}
	We present the highest possible minimum distance of ACD codes on non-symmetric dualities over $\mathbb{F}_4$ in Table \ref{Table 2}. These ACD codes are either optimal or near optimal according to \cite{short_additive}. An $[n,p^k,d]$  additive code is said to be an optimal and near optimal ACD code, if $d=ACD[n,k]$ and $d=ACD[n,k]-1$, respectively, for a given $n$ and $k$. The Proposition \ref{C ACD M and MT} states that if C is an ACD code under duality $M$, then it is also an  ACD code under duality $M^T$.	This implies that the resulting table is the same for both dualities $D_1$ and $D_2$.
	
	 While MAGMA provides built-in functions for determining the dual space of additive codes with respect to standard dualities like (Euclidean, Hermitian and Trace), it lacks functions for the specific dualities  that we are utilizing. Consequently, we have developed an algorithm for non-symmetric dualities over $\mathbb{F}_4$ to assess whether a given additive code meets the criteria for an ACD code. This investigation was conducted using the MAGMA software package \cite{Magma}. For those interested, generator matrices for ACD codes under a non-symmetric duality can be accessed online at \href{https://drive.google.com/file/d/1FsLaVkt8K4k1uQuwVEt5i5iYEqPaIPUj/view?usp=sharing}{https://drive.google.com/file/d/1FsLaVkt8K4k1uQuwVEt5i5iYEqPaIPUj/view?usp=sharing}.
	
	In Table \ref{Table 2}, we identify some new  quaternary ACD codes with optimal parameters under non-symmetric dualities which were not optimal under any symmetric dualities (see \cite{ACDgroupcharacter}). For example, $[5,2^6,6]$, $[6,2^6,4]$, $[7,2^5,5]$, $[7,2^7,4]$, $[7,2^9,3]$, $[8,2^8,4]$, $[8,2^{11},3]$, $[9,2^{10},4]$, $[9,2^{15},2]$, $[10,2^4,8]$, $[10,2^5,7]$ and $[10,2^9,5]$. Moreover, quaternary ACD codes have better minimum distance than quaternary Hermitian LCD codes, see \cite{quaternaryHermitian}. The entries with boldface  represent optimal quaternary codes and $*$ represents the new quaternary codes with optimal parameter.
	\begin{table}[h]
		\caption{Highest minimum distance  for quaternary ACD codes with respect to non-symmetric dualities. The entries with $*$ represent improvements over the prior works.}
		\begin{tabular}{c||cccccccccccccccccccc}
			\hline
			n/k&1&2&3&4&5&6&7&8&9&10&11&12&13&14&15&16&17&18&19&20\\
			\hline
			1&\textbf{1}&\textbf{1}\\
			2&\textbf{1}&\textbf{2}&\textbf{1}&\textbf{1}\\
			3&\textbf{3}&\textbf{3}&\textbf{2}&\textbf{2}&\textbf{1}&\textbf{1}\\
			4&\textbf{3}&\textbf{4}&\textbf{3}&\textbf{3}&\textbf{2}&\textbf{2}&\textbf{1}&\textbf{1}\\
			5&\textbf{5}&\textbf{5}&\textbf{4}&\textbf{4}&\textbf{3}&\bm{$3^*$}&\textbf{2}&\textbf{2}&\textbf{1}&\textbf{1}\\
			6&\textbf{5}&\textbf{6}&\textbf{5}&\textbf{4}&\textbf{4}&\bm{$4^*$}&\textbf{3}&\textbf{2}&\textbf{2}&\textbf{2}&\textbf{1}&\textbf{1}\\
			7&\textbf{7}&\textbf{7}&\textbf{6}&\textbf{5}&\bm{$5^*$}&\textbf{4}&\bm{$4^*$}&\textbf{3}&\bm{$3^*$}&\textbf{2}&\textbf{2}&\textbf{2}&\textbf{1}&\textbf{1}\\
			8&\textbf{7}&\textbf{8}&\textbf{6}&\textbf{6}&$5$&\textbf{5}&\textbf{4}&\bm{$4^*$}&\textbf{3}&\textbf{3}&\bm{$3^*$}&\textbf{2}&\textbf{2}&\textbf{2}&\textbf{1}&\textbf{1}\\
			9&\textbf{9}&\textbf{9}&\textbf{7}&\textbf{7}&\textbf{6}&$5$&\textbf{5}&$4$ &\textbf{4}&\bm{$4^*$}&\textbf{3}&\textbf{3}&$2$&\textbf{2}&\bm{$2^*$}&\textbf{2}&\textbf{1}&\textbf{1}\\
			10&\textbf{9}&\textbf{10}&\textbf{8}&\bm{$8^*$}&\bm{$7^*$}&\textbf{6}&$5$&$5$&\bm{$5^*$}&$4$&\textbf{4}&$3$&\textbf{3}&\textbf{3}&\textbf{2}&\textbf{2}&\textbf{2}&\textbf{2}&\textbf{1}&\textbf{1}\\
			\hline
		\end{tabular}
		\label{Table 2}
	\end{table}
	\section{Conclusion}
	In this article, we have studied the additive codes over finite fields under skew-symmetric dualities. We have given the exact count of symmetric  and skew-symmetric dualities over $\mathbb{F}_q$. We have provided a necessary and sufficient condition for an additive code to be an ACD code. We have given some bounds on the highest minimum distance of an ACD code under skew-symmetric dualities. Moreover, we identified some new optimal codes over $\mathbb{F}_4$ under non-symmetric dualities. It indicates that further study over non-symmetric dualities is worthwhile. Also, one can study additive self-dual codes  over skew-symmetric dualities.
\section*{Acknowledgments}
 The first author is supported by UGC, New Delhi, Govt. of India under grant DEC18-417932. The second author is ConsenSys Blockchain chair professor. He thanks ConsenSys AG for that privilege.	
 \section*{Statements and Declarations}
\textbf{Author’s Contribution}: All authors contributed equally to this work.\\
\textbf{Conflict of interest:} The authors have no conflicts of interest in this paper.

	\bibliographystyle{abbrv}
	\bibliography{Reference}

\end{document}